\documentclass[conference]{IEEEtran}
\IEEEoverridecommandlockouts
\usepackage{hyperref}
\usepackage{mathtools}
\usepackage{cite}
\usepackage{comment}
\usepackage{amsmath,amssymb,amsfonts,amsthm}
\usepackage{authblk}
\usepackage[ruled,vlined,linesnumbered]{algorithm2e}
\usepackage{booktabs}
\usepackage{graphicx}
\usepackage{subcaption}
\usepackage{textcomp}
\usepackage{xcolor}
\usepackage{multirow}
\usepackage{multicol,lipsum}
\usepackage[linesnumbered,ruled,vlined]{algorithm2e}

\SetCommentSty{mycommfont}

\usepackage{wrapfig}

\def\BibTeX{{\rm B\kern-.05em{\sc i\kern-.025em b}\kern-.08em
    T\kern-.1667em\lower.7ex\hbox{E}\kern-.125emX}}

\newtheorem{definition}{Definition}[section]
\newtheorem{theorem}{Theorem}

\begin{document}

\title{Optimal Accuracy-Time Trade-off for Deep Learning Services in Edge Computing Systems
}
% \author{}
\author[1]{Minoo Hosseinzadeh}
\author[1]{Andrew Wachal}
\author[1]{Hana Khamfroush}
\author[2]{Daniel E. Lucani}
\affil[1]{Department of Computer Science, University of Kentucky}
\affil[2]{Department of Engineering, Aarhus University}

\newcommand{\change}[1]{#1}
\newcommand{\nathaniel}[1]{\textcolor{red}{\textbf{NCH:} #1}}
\newcommand{\hana}[1]{{\color{red} \textbf{Hana:} #1}}

\newcommand{\minoo}[1]{{\color{purple} \textbf{Edited:} #1}}

\maketitle

\begin{abstract}
    With the increasing demand for computationally intensive services like deep learning tasks, emerging distributed computing platforms such as edge computing (EC) systems are becoming more popular. Edge computing systems have shown promising results in terms of latency reduction compared to the traditional cloud systems. However, their limited processing capacity imposes a trade-off between the potential latency reduction and the achieved accuracy in computationally-intensive services such as deep learning-based services. In this paper, we focus on finding the optimal accuracy-time trade-off for running deep learning services in a {three-tier EC} platform where several deep learning models with different accuracy levels are available. Specifically, we cast the problem as an Integer Linear Program, where optimal task scheduling decisions are made to maximize overall user satisfaction in terms of accuracy-time trade-off. We prove that our problem is NP-hard and then provide a polynomial constant-time greedy algorithm, called GUS, that is shown to attain near-optimal results. Finally, upon vetting our algorithmic solution through numerical experiments and comparison with a set of heuristics, we deploy it on a test-bed implemented to measure for real-world results. The results of both numerical analysis and real-world implementation show that GUS can outperform the baseline heuristics in terms of the average percentage of satisfied users by a factor of at least $50\%$.
\end{abstract}


\section{Introduction}
\label{sec:Introduction}
{
The promise of \emph{edge computing}~(EC)~\cite{ETSI} has sparked ever-increasing attention in recent years. Much of this attention is a result of ever-increasing consumption and generation of data at the network edge in IoT systems. Because conventional cloud computing relies on a cluster of remote hardware resources, this invites issues for delay-sensitive applications that {are} becoming more prevalent. For instance, self-driving cars that rely on pedestrian detection using complex Deep Learning~(DL) models to perform fast real-time inference has harsh requirements with regard to both provided accuracy (i.e., accurate predictions) and time delay.
% Internet-of-Thing~(IoT) referring to smart objects that are connected together and used to gather and create information about physical environments has created many new challenges in the area of network management. Each device in the IoT system produces a large amount of data that should be processed in real-time to create information about the physical environment. Additionally, emerging applications of the smart cities and IoT era such as automated vehicles and smart health monitoring systems could be both delay sensitive and compute intensive. Although cloud Computing plays a vital role in providing compute resources for different computationally intensive applications, it cannot be a good solution for delay-sensitive applications due to the long distance between the cloud and users. 
% 
%Furthermore, many of these applications usually rely on Deep Learning~(DL) models and as such finding efficient solutions to run DL models is becoming a popular research trend. 
% To bring the central cloud benefits near to the users and reduce the task completion latency, Edge Computing~(EC)~\cite{ETSI} technology has been recently introduced. 
However, because compute resources deployed for EC are less powerful than that of the cloud computing, it is necessary to consider the trade-off with regard to computation, communication, and storage capacities of the edge servers. Much recent works have investigated strategies for adhering to these resource limitations while providing ample Quality-of-Service~(QoS)~\cite{abbas2017mobile,he2018s, farhadi2019service,turner2020meeting}.
% Although EC provides a lot of benefits for End Users~(EU), it comes with limited computation, communication, and storage capacities. Therefore, many of the recent works investigated optimal use of a multi- of cloud and edge computing to provide services for such applications~\cite{mao2017survey}.
%Furthermore, many of the emerging IoT applications rely on Deep Learning~(DL) models that require large computing resources.
%Therefore, resource management plays a key role in MEC system which has been investigated in many papers~\cite{mao2017survey}. 
% For example, works in~\cite{abbas2017mobile,he2018s, farhadi2019service,turner2020meeting} have investigated different resource management strategies for increasing the Quality of Service (QoS) in a two- edge and cloud computing system for a general set of service types.

%The QoS defined in previous works usually focus on one single measure of QoS, such as time.
Many emerging IoT applications rely on deep learning technology to infer knowledge from the data and make smart decisions. More recently, EC has been shown to reduce response time for deployed DL services~\cite{wang2020convergence,felemban2019picsys,wang2019edge}.
However, DL {models} are computationally expensive to run, especially for resource-constrained IoT devices. Thus, there is a need for optimal offloading strategies for handling requests for DL-based services in EC.
% implementing and using DL models on a multi- edge and cloud computing systems have to be proposed. 
% 
%Deep learning models as a service for different applications in IoT concept is getting more popular every day. Many IoT applications are dependent on DL models. 
% 
Furthermore, response/delay time is not the only crucial aspect of QoS for DL-based applications. Accuracy/loss of the provided DL-based services is also an important measure of the QoS. There is often a trade-off with regard to these two aspects of QoS. For instance, more costly models often take longer to run.
% Accuracy increase and completion time decrease of a DL task are usually conflicting, since a more accurate DL model may require more time for both training and inference. 
% 
%they might also want to have a good level of accuracy/loss which is dependent on the DL model. 
%For instance in a theft detection application for home safety, the quality of the theft detection service is not only dependent on the delay of detecting a  theft, but also how accurate the detection results are. 
%In fact, a higher level of provided accuracy for a DL model usually comes with more intense computation and more delay.
%depends on the In these new categories of applications, a reasonable level requires new definition of QoS/QoE. Although providing a more accurate DL model would provide a better service in terms of accuracy, on the other hand it will increase the computation time. 
Thus, it is crucial that we have optimal accuracy-time trade-off for DL-based services in complex EC platforms.
% while considering the network resource limitations is of utmost importance for DL-based services in a multi-edge and cloud computing platform.

% \begin{figure}
%     \centering
%     \includegraphics[width=1.0\linewidth]{Edge-Cloud.pdf}
%     \caption{The three- user-edge-cloud architecture}
%     \label{fig:MEC Architecture}
% \end{figure}

%% Mobile edge computing resource management
The challenge of maximally meeting user QoS expectations via offloading strategies in the edge computing systems is well-investigated in the literature. 
% There are several parameters that must be considered when modeling these systems and their provided QoS.
% One key process that affects users' experience is resource management.
%Several papers addressed resource management in MEC with different directions~\cite{mao2017survey}. Resource management can be done centrally on the cloud server~\cite{you2016energy,he2016service} or on each MEC server in a distributed manner~\cite{mach2017mobile}.
Several studies explored task offloading decisions in EC systems using different objectives, such as, minimizing task completion time~\cite{liu2016delay}; minimizing both latency and chances of application failure~\cite{mao2016dynamic}; and minimizing both End Users~(EU) energy consumption and task's completion time~\cite{kamoun2015joint}.
%MEC accelerates response time for computational intensive services, e.g., ML tasks \cite{wang2020convergence,felemban2019picsys}.
Other works considered DL model accuracy while providing a service considering different criteria, such as energy consumption, computation, and network condition~\cite{wang2020convergence,felemban2019picsys}.
% Although the state-of-the-art have separately studied implementation of some specific learning techniques~(e.g., DL) on EC systems, 
However, there still exists a large gap in understanding the trade-offs between different conflicting QoS metrics, while making request scheduling/offloading decisions for DL applications, particularly in multi-edge and cloud systems.
In this paper, we study optimal QoS-aware task offloading strategies for DL-based services in a three-tier user-edge-cloud computing platform. We cast the problem as an \emph{integer linear programming}~(ILP) problem where the constraints are inspired by the limited hardware resources available in the layer of edge clouds. We prove that our problem is NP-hard and propose an efficient greedy algorithm that provides close-to-optimal performance. We evaluate the proposed algorithm's performance w.r.t. the trade-off between provided accuracy and time delay from offloading decisions and compare its performance against several baselines. We perform this evaluation in both numerical simulations and in our own real-world edge computing test-bed. To the best of our knowledge, this is the first work to explicitly consider such trade-offs for offloading decisions in a three-tier user-edge-cloud framework. 

% To fill some of these gaps, in this paper we focus on finding optimal task offloading strategies for a set of DL-based services in a three tier user-edge-cloud computing platform while considering the limited capacities of edge servers and two conflicting QoS metrics; namely accuracy and completion time of the DL-based service. 
%This framework is shown in Fig.~\ref{fig:MEC Architecture}.
% We propose an Integer Linear Programming model for the offloading problem considering both the accuracy of DL model and the completion time of the task. The proposed model is proved to be NP-hard, and so we provide a greedy algorithm to find an efficient close-to-optimal solution. We evaluate the proposed algorithm's performance in terms of the accuracy-time trade-off and in comparison with multiple baseline solutions through both numerical analysis and real-world test-bed implementations. To account for the accuracy-time trade-off, a new metric of user satisfaction is proposed and is tested under several different network setup conditions. 
% To the best of our knowledge, this is the first work that explicitly considers such trade-offs in offloading decisions of a multi-tier user-edge-cloud framework.
%In our proposed model, we assume that each user will be satisfied if and only if both \emph{(i)} the accuracy of a provided model for its request is greater than or equal to what they requested and \emph{(ii)} the provided request completion time less or equal to what they requested.

}%% End of Introduction
%%%%%%%%%%%%%%%%%%%%%%%%%%%%%%%%%%%%%%%%%%%%%%%%%%%%%%%%%%%%%%%%
\section{Problem Formulation}
\label{sec: Problem Definition}
% Here, we formally describe our EC environment architecture that we use in addition to the formal problem definition we solve. We assume a three-tier user-edge-cloud architecture.
Here, we formally describe our considered EC architecture, as well as formally define our offloading problem. We consider a three-tier architecture.
The topmost tier represents cloud servers that can be physically distributed, but virtually centralized. Here, we consider them as one single entity, called a ``cloud server". The middle layer consists of a heterogeneous set of edge computing servers which are connected to the cloud server. We also consider that the edge servers can also directly communicate with each other, either through a back-haul or a device-to-device communication. Users represent the bottom layer in this architecture and we assume users can request different computationally-intensive services from their local EC servers. A local EC server of a user is the one that is the closest edge server to the user. 
% In our experimental results, 
For our work, we focus on service types that use a DL model to serve the user. However, the model is general to any service type. Different service types are used in this system, not all services can be placed on each edge server due to the edge device capacity constraints, but all service types could be placed on the central cloud server, since the cloud resources are much larger than those of the edge servers. Simply, we assume the central cloud has communication and computation constraints, but no storage constraints.
% no storage constraints for the central cloud, however, we consider communication and computation capacity constraints on the central cloud.
% Without loss of generality, 
Service types represent different DL-based tasks (e.g., image classification). We consider a set $K$ of services available in the system and a set $L$ of DL model types per service. Each service type $k\in K$ has $|L|$ different implementations of a DL model that could be used to serve the task with different levels of provided accuracy. %\minoo{Each model $l \in L$ has a different level of accuracy.}

\textbf{Model description.} 
%For the edge-to-cloud platform in Fig.~\ref{fig:MEC Architecture}, 
We consider a three-tier user-edge-cloud system and a set of requests~$N$ submitted to the system by the users at a given time. Each request $i \in N$ is submitted alongside a preferred accuracy and delay thresholds to be considered for that request's respective user satisfaction. A user with several requests can be modeled as multiple users with one request each. Note, we use ``request" and ``user" interchangeably throughout the paper. We consider a set of servers~$M$ that can serve user requests. Each server $j\in M$ has known computation and communication capacities, represented by $\gamma_j$ and $\eta_j$, respectively.
%Under our model, we consider cloud servers and edge servers to all have capacities for both communication and computation resources, $\eta_j,\gamma_j\;(\forall j\in M)$, respectively. 
Many papers modeling EC consider storage capacity constraints; we do not consider storage constraints in this paper as the assumption is that service/model placement decision is already made. Unlike most state-of-the-art papers, we assume the cloud has limited resources --- albeit its resources are more powerful than the edge servers.
% that assume unlimited resources for the cloud server, we assume that the cloud is clearly more powerful than the edge servers, however, it also has limited resources. 
This is because the scale of the edge computing systems in real-world scenarios are very diverse. Cloud servers could be devices like desktop computers interacting with more resource-constrained edge devices than a large-scale remote cloud with seemingly unlimited resources. Thus, we directly and explicitly consider the resource capacities of the cloud server to allow our model to be more general and more realistic. Additionally, our approach \emph{allows for the} consideration of more than one cloud server in the topmost layer.

For our problem, we do not explicitly distinguish between edge servers and cloud servers because both can be modelled in the same way. The main difference is on the allocated hardware resources, i.e., cloud servers have significantly more resources. Also users cannot communicate with the cloud directly, and they can only communicate with the cloud through their local edge server. Thus, there is no user request submitted to the cloud directly.
It is assumed that at the start of each time-step, all servers $j\in M$
report their remaining computation and communication capacities along with the user requests they received. 
%To find the optimal solution, we assume that there exists a centralized entity that can fully observe the system and share thr observations with all edge . 
Therefore, the available capacity of all servers $j \in M$ and the set of user requests is assumed to be known by each server at every given time step. For simplicity, we do not consider the overhead of sending such control information in the system, as this is out of the scope of this paper. It is assumed that the user requests are submitted to their local edge servers. 

% As a reminder, we assume that the services are already placed and 
The goal of our problem is to make optimal request scheduling decisions, i.e. which server should serve a given request sent by a user to the user's local edge server. Note that for every request submitted to an edge server, one of the following scheduling decisions can be made: (a) the edge server serves the request locally, (b) the edge server offloads the received request to another available server (this includes both central cloud or other edge servers), or (c) the edge server drops the request entirely.
At the beginning of each time-step, each edge server makes a scheduling decision regarding each service request submitted by its associated users. The edge server covering (i.e., directly associated with) user $i$ is shown by $s_i$, and it is the edge server which receives user $i$'s request at the first place.
%assume that each user $i$ has a unique associated edge server represented by $s_i$, that is essentially the closest edge server to user $i$ and the edge server 
%chooses whether (a) to locally schedule the service request; (b) to offload it to another available server; or (c) to drop it entirely. 
%\hana{I rewrote $s_i$'s definition in a sentence right before this point. why did you define $I$ here? Also definition of I is not clear. Try to rewrite its definition and move it to after equation 2(f)}\minoo{Hana I changed it a bit and moved it after equation 2(f)}
%For our model, we consider services to be implemented by different DL models with different accuracy. 
% We consider the decision of placing instances of a DL models for the set of services to have already been made. 

To review, when a user submits a request $i$, they are requesting some service $k\in K$. They also submit minimum required accuracy $A_i$, and maximum tolerable completion time $C_i$ to satisfy request~$i$.
A user will be satisfied only in the condition that the DL model scheduled to serve it fulfills both of its requested accuracy and completion time requirements.

%Under our model, we consider cloud servers and edge servers to all have capacities for both communication and computation resources, $\eta_j,\gamma_j\;(\forall j\in M)$, respectively. Many papers modeling EC consider storage capacities; we do not consider storage constraints, as the assumption is that service/model placement decision is already made. Moreover, many papers studying EC will consider the cloud server to have unlimited resources. We do not because the scale of systems considered to be EC are diverse. Cloud servers could be smaller devices like desktop computer interacting with more resource-constrained edge devices than a large-scale remote cloud with seemingly unlimited resources. Thus, we directly and explicitly consider the resource capacities of the cloud server to allow our model to be more general. Additionally, our approach allows to consider more than one cloud server in the topmost layer.

\textbf{Completion time.} Users' requests experience different delays based on the decision to offload or process locally. Each request may experience the following types of delay throughout the system: \emph{queuing delay}, \emph{processing delay}, and \emph{communication delay}. 
There usually exist two queues for each edge server: one for admission control (accepting requests into the system) and one for processing the services already admitted by an edge server.
For this work, we assume that the delay due to the second queue is negligible, however, we explicitly consider the admission control queuing delay in our problem formulation. Since the focus of this work is on finding the optimal request scheduling decision at the edge servers, we ignore the communication delay created by sending the requests (and their dependencies) from users to the edge servers and receiving the results from the edge servers, as this delay will be the same for any decision made~(e.g., offloading vs. local processing) at the edge servers. We will also assume that the delay caused by running decision algorithms at the edge servers is negligible. We will discuss the running time of our proposed decision making algorithm in Section.~\ref{sec:Proposed Algorithm}. 
Now we formally define each type of delay: \emph{\textbf{1) Queuing Delay ($T^{q}_{ij}$)}} 
is the time between the arrival of request $i$ to edge server $j$ and when a decision regarding how to process the request is made on that edge server. We assume a time-slotted scenario and
at each time slot, the requests that arrived at edge server $j$ will be held in a queue until a decision is made at the end of a time frame. Each time frame consists of multiple time slots. 
%%NEXT SENTENCES ARE REDUNDANT
%To reduce the queuing delay and avoid extra waiting time in case the request arrival rate is low, for our real-world implementation, we also assume that a decision is made either when the queue is full or at the end of a time frame whichever is shorter. This assumption is released for our numerical analysis as we set the requests arrival rate in a way that a queue is always full at the end of each time frame. We explained this part in the testbed. So this is redundant
\emph{\textbf{2) Processing Delay ($T^{proc}_{ijkl}$)}} 
is the time that will take for the request $i$ to be processed on server.
It depends on the service type $k$, DL model $l$ used for that service, and the processing power of edge/cloud server $j$~(CPU, GPU, RAM) processing the request.
\emph{\textbf{3) Communication Delay ($T^{comm}_{ijj'}$)}} 
is the communication delay needed to send request $i$ from edge server $j$ to another server $j'$ in case that an offloading decision is made. $T^{comm}_{ijj'}$ is computed based on our testbed results which is discussed in section~\ref{sec:Numerical Results}.
Now, we define completion time of serving request $i$ at the $j$-th server for service type $k$ using DL model type $l$ as following: if the the process is offloaded, the completion time will be calculated as $c_{ijkl} = T^{comm}_{is_{i}j} + T^{q}_{is_{i}} + T^{proc}_{ijkl}$ and if it is done locally, the completion time will be computed as $c_{ijkl} = T^{q}_{is_{i}} + T^{proc}_{ijkl}$, {where $s_i$ represents the covering edge server of request $i$ as noted before}.
% \begin{equation}
%     \small
%     %\begin{split}
%     c_{ijkl} = 
%     \begin{cases}
%         T^{comm}_{is_{i}j} + T^{q}_{is_{i}} + T^{proc}_{ijkl}  & \text{if $s_i \neq j$ } \\
%         T^{q}_{is_{i}} + T^{proc}_{ijkl}  & \text{otherwise}
%     \end{cases}
%     %c_{ijkl} =  T^{comm}_{is_{i}j} + T^{q}_{is_{i}} + T^{proc}_{ijkl} + T^{pre}_{ijkl}
%     %\end{split}
%     \label{eq:ct}
% \end{equation}
\\
{\textbf{MUS problem definition.}} Now, we formally define our proposed problem that we call the \textit{Maximal User Satisfaction~(MUS)} on the Edge problem. The goal of this problem is to maximize the expected user satisfaction based on the QoS requirements submitted alongside each service request. We cast the MUS problem as an {ILP} problem. First, we introduce our definition of user satisfaction as a combination of both the {\emph{expected accuracy provided by the DL model}} and the \emph{total completion time} of the request.

\begin{definition}
    (User Satisfaction (US)). 
    We consider a user to be satisfied if, and only if, both the provided accuracy is equal or more than the user's requested accuracy and the provided completion time is equal or less than the user's requested delay. We formally define the US function for request $i$ as following:
    %More formally, we propose the US for the request $i$ as a {combination of both provided completion time and the accuracy} as following:
    % \vspace{-15mm}
    \begin{equation}
        US_{ijkl} =  w_{ai} \Big(\frac{{a_{ijkl} - A_i}}{Max_{as}}\Big) + 
                w_{ci} \Big(\frac {{C_i - c_{ijkl}}}{Max_{cs}}\Big)
    \end{equation}
   % More formally, we propose the US for the request $i$ as following:
    %\vspace{-20mm}
%     \begin{equation}
%     \label{Def:US}
%     \small
%     \begin{aligned}
%     \alpha_{ijkl} = w_{ai} \Big(\frac{{a_{ijkl} - A_i}}{Max_{as}}\Big) \;
%     \beta_{ijkl} = w_{ci} \Big(\frac {{C_i - c_{ijkl}}}{Max_{cs}}\Big)\\
%     US_{ijkl} = 
%     \begin{cases}
%         \frac{1}{(1 + e^{-(\alpha_{ijkl} \; \beta_{ijkl})})} & \text{if $\alpha_{ijkl} \geq 0 \; \& \; \beta_{ijkl} \geq 0$ } \\
%         \frac{-1}{( 1 + e^{-(\alpha_{ijkl} \; \beta_{ijkl})})} & \text{otherwise}
%     \end{cases}
%     \end{aligned}
% \end{equation}
%\vspace{-20mm}
\normalsize
    \noindent where $US_{ijkl}$ is the user satisfaction provided by serving request $i$ at the $j$-th \change{server} using DL model type $l$ of the requested service type $k$. $C_i$ and $A_i$ are the completion time and accuracy thresholds asked by user for request $i$. $a_{ijkl}$ and $c_{ijkl}$ represent the accuracy and the completion time of serving request $i$ at server $j$ using service type $k$ and model type $l$, respectively. 
    %\minoo{The $a_{ijkl}$ is dependent on the DL model $l$ placed on server $j$ for service type $k$.} 
    $Max_{as}$ is the maximum possible provided accuracy in the system and $Max_{cs}$ is the maximum (worst case) completion time of a task in the system. It is assumed that these values are known. Additionally, $0\leq w_{ci}\leq 1$ and $0\leq w_{ai}\leq 1$ are the weights that the user would assign to the requested delay and requested accuracy, respectively. These weights can be used to model different importance levels and priorities for accuracy and task's completion time. For example, some users may care more about getting more accurate results and they can tolerate some levels of delay. 
    %The reason that we modified \textit{Sigmoid} function for US definition is that the final US of each user in our definition will be between negative one and one. Additionally, if user get satisfied it would be greater than or equal to $0.5$ (both $\alpha$ and $\beta$ are positive or zero in this case).
%     {The complete notations are listed in Table~\ref{table:notation}.}
\end{definition}
%Now, we introduce the formal ILP definition for our problem. 
Our problem considers a set of binary decision variable $\textbf{X}\triangleq(X_{ijkl})\;{\forall i\in N, j\in M, k\in K, l\in L}$ where $X_{ijkl}=1$ if and only if request $i$ is served by device $j$ using service $k$ and DL model $l$, 0 otherwise. Below is our proposed ILP formulation:

\begin{subequations}
    \small
    \label{eq:MUS}
    \begin{align}
        \text{Max: }   
            & \frac{1}{|N|} \Big(\sum_{i\in N,j\in M} \sum_{k\in K,l\in L} {US_{ijkl} X_{ijkl} }\Big) 
            \tag{\ref{eq:MUS}}
        \\
        \text{s.t.: } 
            &  \sum_{j\in M,k\in K,l\in L} {X_{ijkl} \leq 1}, 
            \qquad \forall i\in N,
            \label{eq:1}
            \\
            & {{X_{ijkl}\; {a}_{ijkl}} \geq {X_{ijkl}\;A_{i}}}, 
            \qquad \forall i\in N,j\in M,k\in K,l\in L,
            \label{eq:accuracy}
            \\
            & {X_{ijkl}\; {c}_{ijkl} \leq {X_{ijkl}\;{C}_{i}}}, 
            \qquad \forall i\in N,j\in M,k\in K,l\in L,
            \label{eq:delay}
            \\
            & \sum_{i,k,l} {X_{ijkl}\;{v_{ijkl}}\leq \gamma_j}, 
            \qquad \quad \forall j\in M,
            \label{eq:computation}
            \\
            & {\sum_{i,k,l}\sum_{j'\neq j}{I_{ij}\; X_{ij'kl} }\;{u_{ijkl}}\leq \eta_j}, 
            \qquad \quad\forall j\in M,
            \label{eq:communication}
            \\
            & X_{ijkl} \in {\{ 0,1\}}, 
           \qquad  \forall i\in N,j\in M,k\in K,l\in L,
            \label{eq:decVariables}
    \end{align}
\end{subequations}

\normalsize
Where,{
%Each request $i\in N$ is directly associated with one single edge server $j \in M$ s.t. $j > 0$, or that an edge server directly \emph{covers} a set of requests which we call it $s_i$. 
$I$ is an indicator vector such that $|I|= |N|\times |M|$ and $I_{ij}=1\;(\forall i\in N, j\in M)$ if and only if (user) request $i$ is directly covered by edge device $j$, and $0$ otherwise (i.e. $I_{ij}=1$ where $j=s_i$).} $u_{ijkl}$ and $v_{ijkl}$ are communication cost and computation cost of serving request $i$ on server $j$ for service type $k$ using model type $l$, respectively; and $\gamma_j$ and $\eta_j$ are the communication and computation capacity of server $j$. 
%\minoo{$j'$ is the index of edge/cloud server which is going to serve request $i$ in the offloading case in Eq.~\eqref{eq:communication}}. 
The first constraint guarantees that each request should be served in just one server using one service and one DL model in Eq.~\eqref{eq:1}, or it will be dropped. The Eq.~\eqref{eq:accuracy} ensures that if the request $i$ is going to be served by the $j$-th device, the accuracy of the result should be more than or equal to the requested accuracy by the user. The Eq.~\eqref{eq:delay} guarantees that the completion time of serving the request $i$ has to be less than or equal to user's requested delay for request $i$. Constraints ~\eqref{eq:computation} and ~\eqref{eq:communication} ensure that the total computation or communication costs needed to process or offload all requests coming to device $j$ must not be more than the overall computation and communication capacity of that device, respectively. In this optimization problem, there exist three possible scheduling choices: 1) \textit{Local processing.} The request will be served on the edge server; 2) \textit{Offloading.} The request will be offloaded to {either} cloud server or one of the neighboring edge servers; 3) \textit{Drop.} The request will not be served and will be dropped.\\
\textbf{Special case.} Although the proposed formulation considers strict QoS requirements, we can however define other cases as a special case of this problem. For instance, by relaxing constraints (2b) and (2c), our problem can model scenarios where users QoS requirements are more of a suggestion rather than a hard constraint. This means that a user can be served even if its QoS requirements are not strictly met. 
%We can use our problem formulation for less strict QoS conditions. 
%\textbf{Hardness of MUS Problem}
\begin{theorem}
    The proposed MUS problem is NP-hard. 
    \label{theorem:np}
\end{theorem}

\begin{proof}
    We prove Theorem~\ref{theorem:np} by a reduction from the NP-hard Maximum Cardinality Bin Packing~(MCBP) problem to our problem.
    We are given $m$ bins of identical capacity $C$ and a set of $N = \{1,2,...,n\}$ items of weights (size) $p_{i} (i=1,...,n)$. The objective is to maximize the number of items packed into the $m$ bins without exceeding bin capacities and without splitting items\cite{loh2009solving}. The decision making variable is defined as $x_{ij}=1$ if item $i$ is placed into bin $j$, and zero otherwise.
    We will show that a simple instance of our problem would be as hard as the MCBP problem. We now construct an instance of our problem. 
    For each item $i \in n$ construct a request $r_i$ with total cost of getting satisfied equal to \textbf{$ u_{i} + v_{i} = p_i$}, given that the communication and computation costs of serving request $i$ on any server will be equal, i.e., $u_{i}=u_{ijkl},v_{i}=v_{ijkl}, \forall j,k,l$. We also construct $m$ servers with $m$ bins.
    We consider that all edge servers and cloud servers have identical capacities equal to $C$ equal to $\gamma_j + \eta_j = C\; \forall j$. The goal is to maximize the number of items/requests which can be placed into the $m$ edge servers/bins of identical capacities.
    We claim that the optimal solution to the constructed instance of our problem gives the optimal solution to the MCBP problem. This is because an algorithm that solves our problem can solve the MCBP. Since MCBP is NP-hard, this concludes our proof~\cite{chen2015efficient}.
\end{proof}

\section{Proposed Greedy Algorithm}
\label{sec:Proposed Algorithm}
Since MUS problem is NP-Hard (Theorem~\ref{theorem:np}), we propose a greedy algorithm which we call \emph{GUS} that attains near-optimal results w.r.t. to the objective function defined in Eq.~\eqref{eq:MUS}.
At each round, a request is served using the following rule: The GUS algorithm calculates the US of each \textit{candidate server} which has service $k$ and model type $l$ for serving request $i$, and has enough computation capacity to serve this request, and is able to satisfy the request's minimum requirements to be a candidate server; meaning that the expected accuracy of its model should be greater than or equal to the requested accuracy and the expected total completion time provided by this server at the time of decision should be less than or equal to the requested delay. The best server with the highest US which has enough capacity to serve request $i$ will be then selected. If request $i$ is going to be offloaded, the communication capacity of the server that covers request $i$~($s_i$) should be also available. If there is enough capacity, the request will be assigned to server $j$; otherwise the algorithm will check next best candidate server with next highest US. If there is no server which can satisfy the request, it will be dropped. At the end of each round, the algorithm updates the remaining capacity for each edge cloud. The algorithm will repeat this process for all the requests. The pseudo-code for the proposed GUS algorithm is provided in Algorithm~\ref{alg:GO}.
%\textbf{Discussion on Computational Complexity. }
The complexity of finding optimal solution for the MUS problem in the worst case is of $O((|L||M|)^{|N|})$.
The running time of the GUS algorithm in the worst case is $O(|N|((|L||M|)^2 + |M| + |L||M|)) = O(|N|((|L||M|)^2))$.

\begin{algorithm}[t]
    \SetAlFnt{\footnotesize}
    \SetAlCapFnt{\small}
    \SetAlCapNameFnt{\small}
    \DontPrintSemicolon
    \SetNoFillComment
    \SetKwInOut{Input}{Input}\SetKwInOut{Output}{Output}
    \Input{{Given $N$, $M$, $K$, $L$, $I$, $Max_{as}$, $Max_{cs}$, $A$, $C$}}
    \Output{Find $X_{ijkl}$ for each request $i$}
    \ForEach{request $i \in$ Requests}
    {
        $s_i$ = $\{j\;|\;I_{ij} = 1\}$\;
        \ForEach{server $j$ $\in$ sorted servers having service $k$ based on higher US}
        {
            \If{$c_{ijkl}\leq C_i$ and $a_{ijkl} \geq A_i$ and $v_{ijkl} \leq \gamma_j $}
                {
                \If{$j$ == $s_i$}
                    {
                    $X_{ijkl}$ = 1\;
                    Locally process request $i$\;
                    % status[i] = locally processing\;
                    % serving\_place[i] = j\;
                    % service[i] = k\;
                    % model[i] = l\;
                    update $\gamma_j$\;
                    break\;}
                % \ElseIf{$j\not=s_i$ and $ u_{ijkl} \leq \eta_{s_i}$}
                \ElseIf{$u_{ijkl} \leq \eta_{s_i}$}
                    {
                    $X_{ijkl}$ = 1\;
                    Offload request $i$\ to place $j$\;
                    % status[i] = offloading\;
                    % serving\_place[i] = j\;
                    % service[i] = k\;
                    % model[i] = l\;
                    update $\gamma_j$ and $ \eta_j$\;
                    break\;}
                }
            }
    }
    \caption{Proposed Greedy Algorithm (GUS)}
    \label{alg:GO}
\end{algorithm}

\section{Results}
\label{sec:Numerical Results}
\begin{figure*}[t!]
\begin{subfigure}{0.25\textwidth}
        \includegraphics[width=\linewidth]{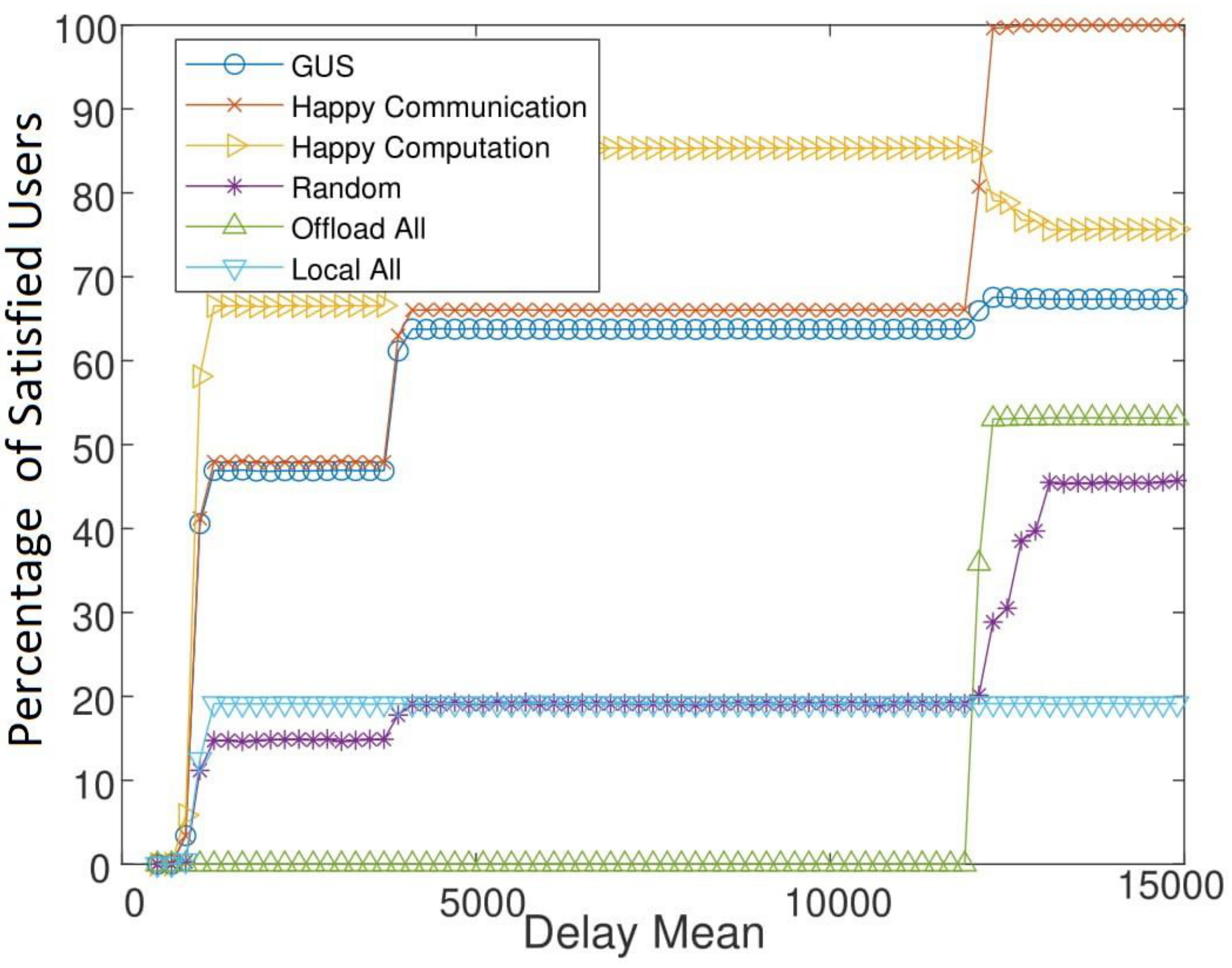}
        \caption{} \label{fig:Del-US-100}
    \end{subfigure}\hspace*{\fill}
    \begin{subfigure}{0.25\textwidth}
        \includegraphics[width=\linewidth]{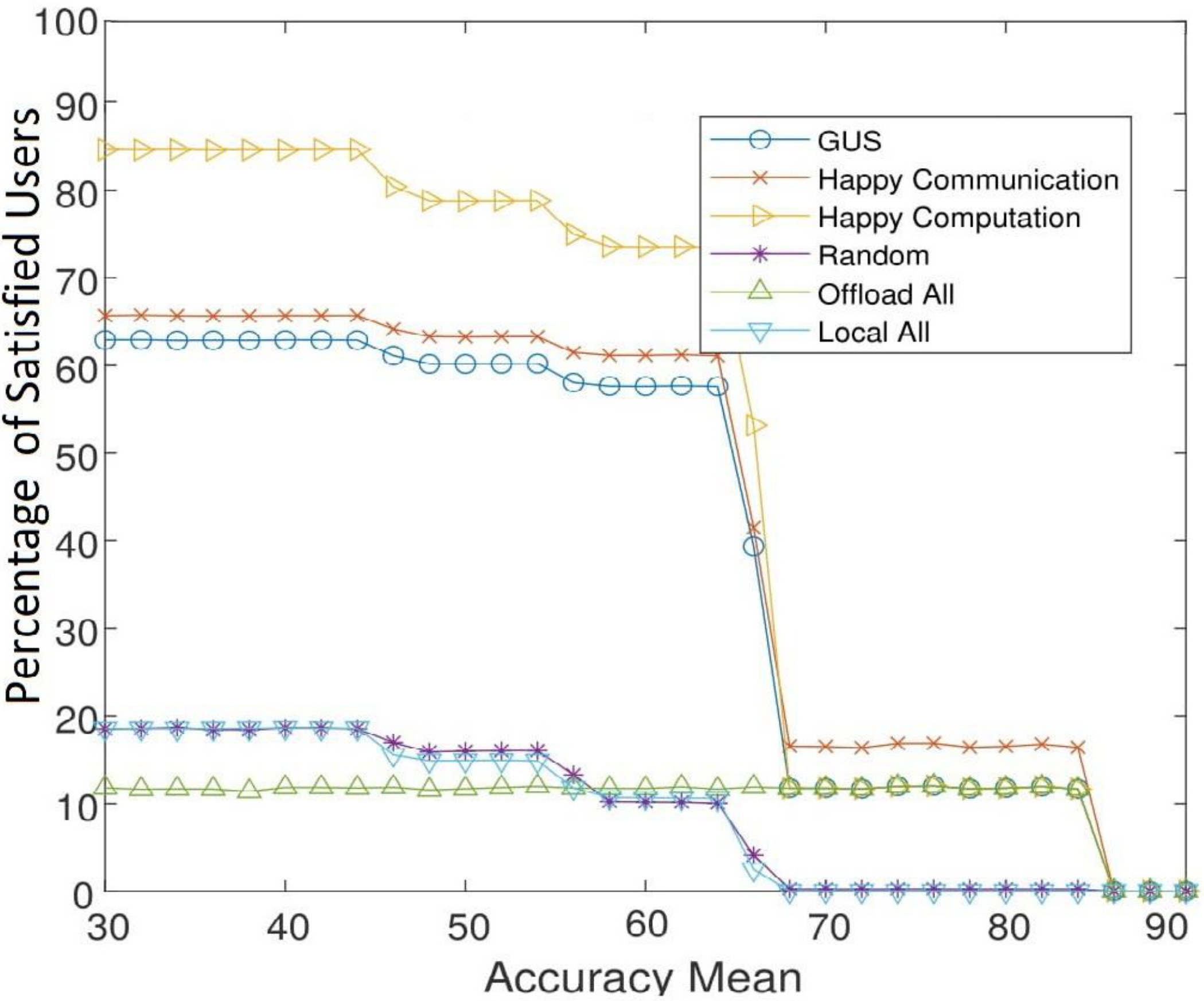}
        \caption{} \label{fig:Acc-US-100}
    \end{subfigure}\hspace*{\fill}
    \begin{subfigure}{0.25\textwidth}
        \includegraphics[width=\linewidth]{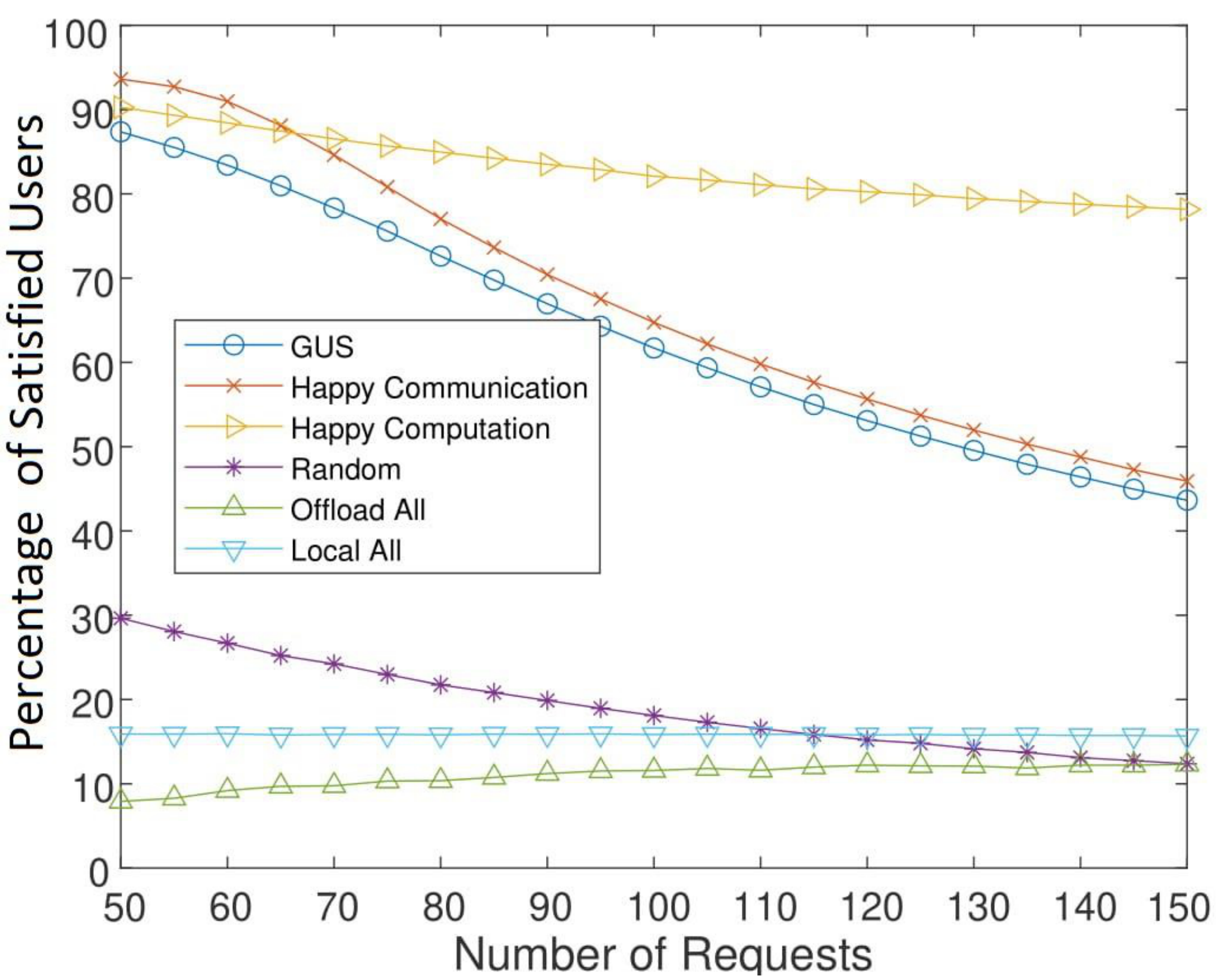}
        \caption{} \label{fig:n-US}
    \end{subfigure}\hspace*{\fill}
    \begin{subfigure}{0.25\textwidth}
        \includegraphics[width=\linewidth]{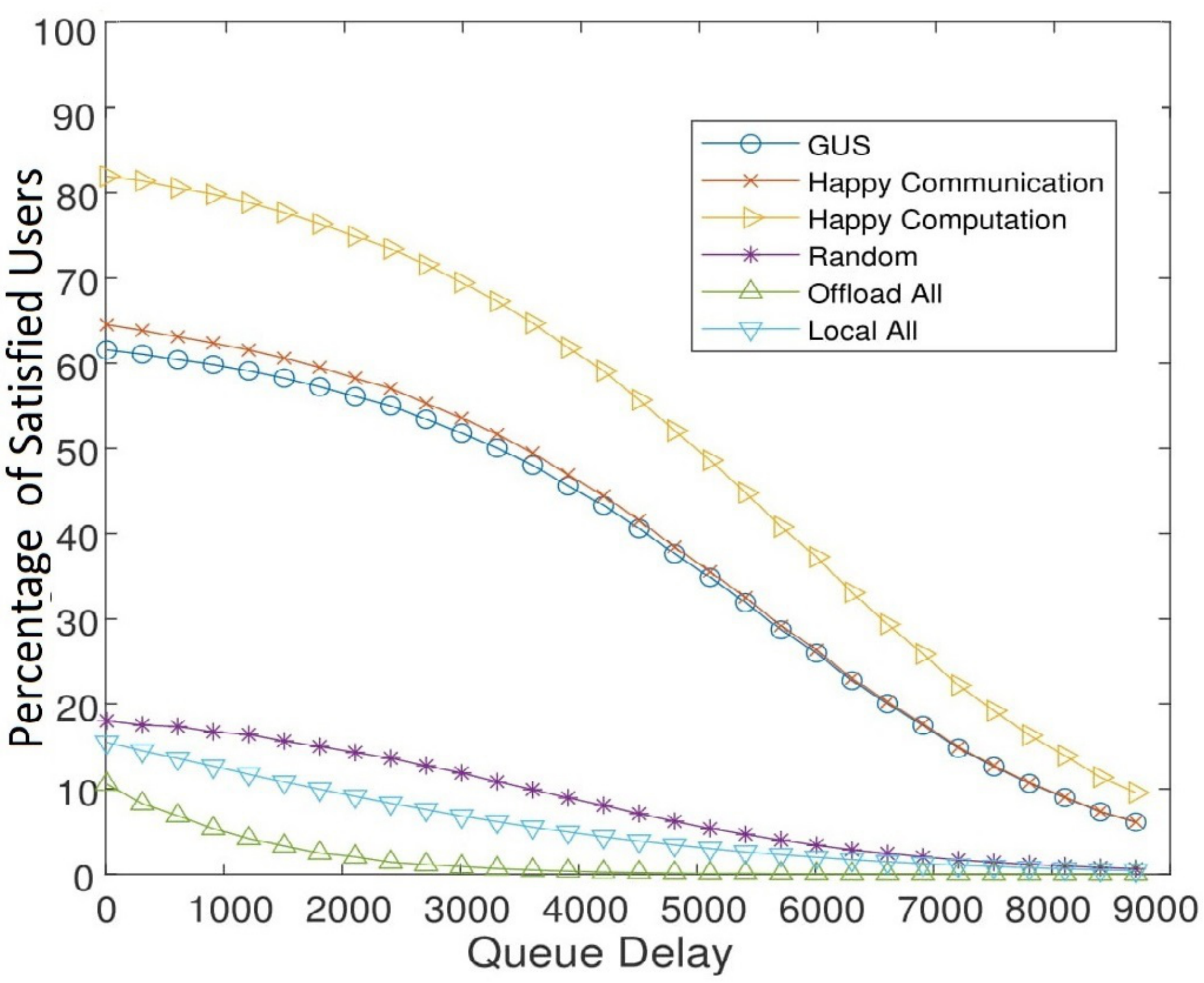}
        \caption{} \label{fig:Queue-US-100}
    \end{subfigure}
    %\caption{Numerical Results} \label{fig:Testbed Results}
    \medskip
    \begin{subfigure}{0.245\textwidth}
        \includegraphics[width=\linewidth]{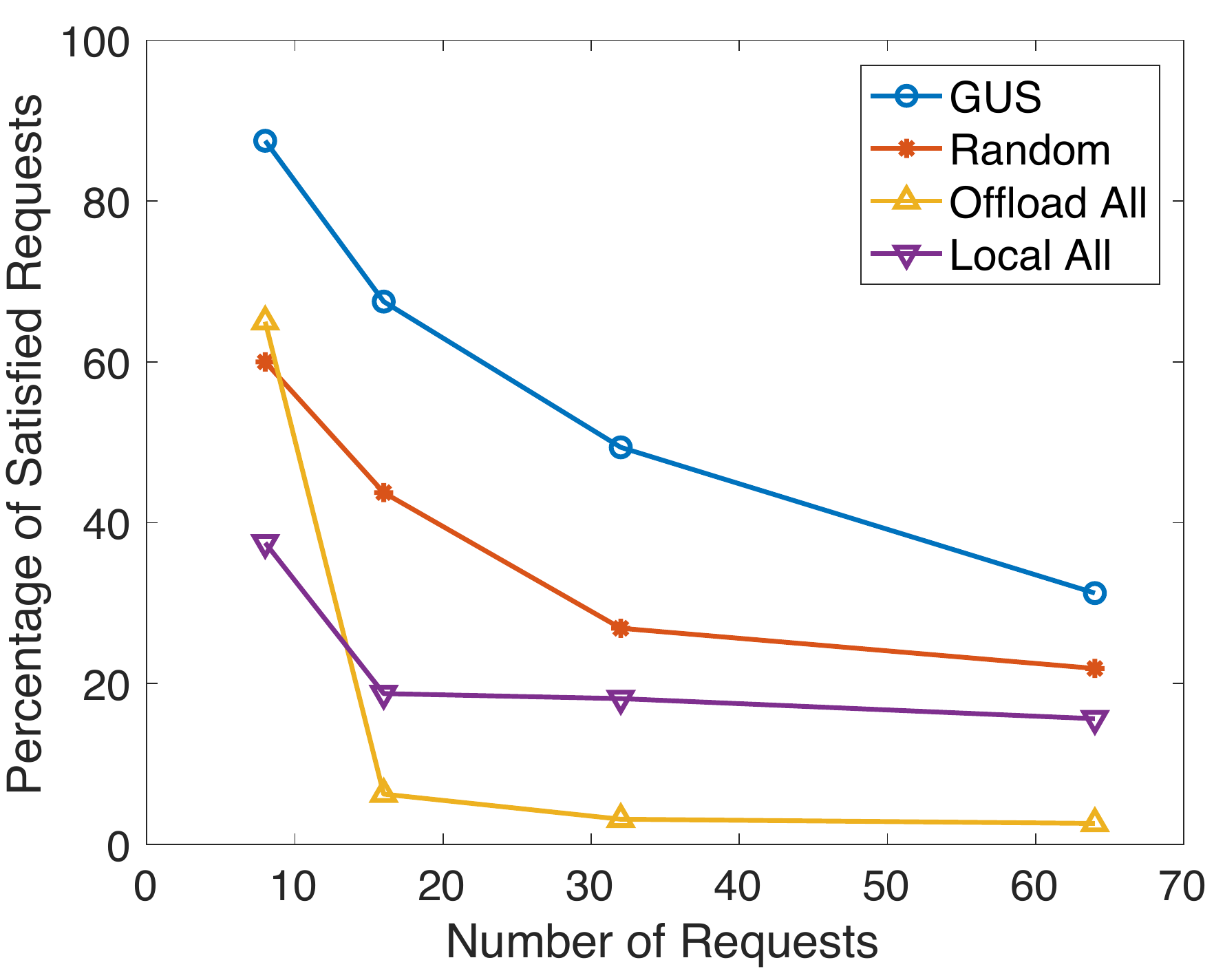}
        \caption{} \label{fig:Del-US-100-2}
    \end{subfigure}\hspace*{\fill}
    \begin{subfigure}{0.245\textwidth}
        \includegraphics[width=\linewidth]{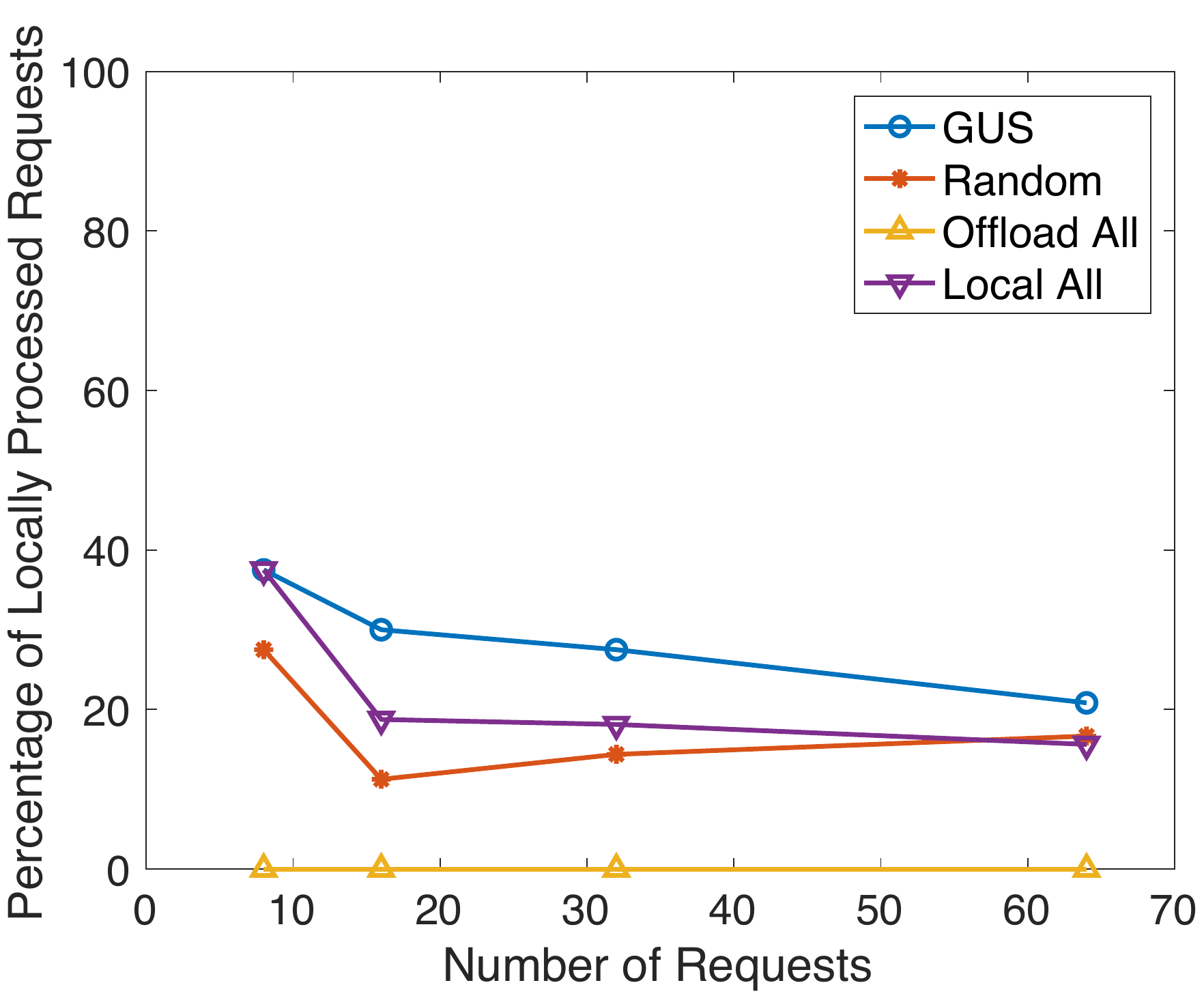}
        \caption{} \label{fig:Del-US-100-3}
    \end{subfigure}\hspace*{\fill}
    \begin{subfigure}{0.245\textwidth}
        \includegraphics[width=\linewidth]{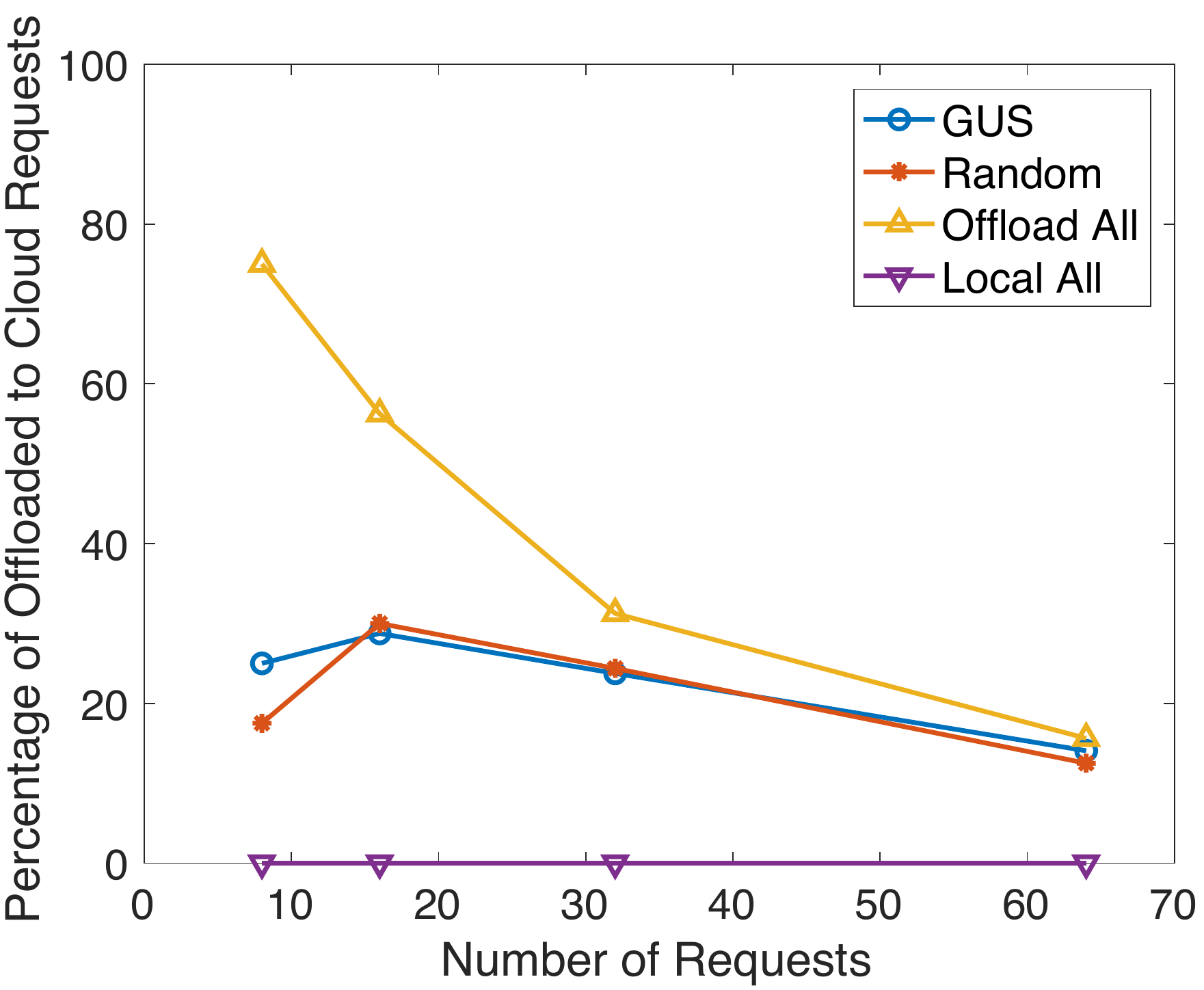}
        \caption{} \label{fig:Del-US-100-4}
    \end{subfigure}\hspace*{\fill}
    \begin{subfigure}{0.245\textwidth}
        \includegraphics[width=\linewidth]{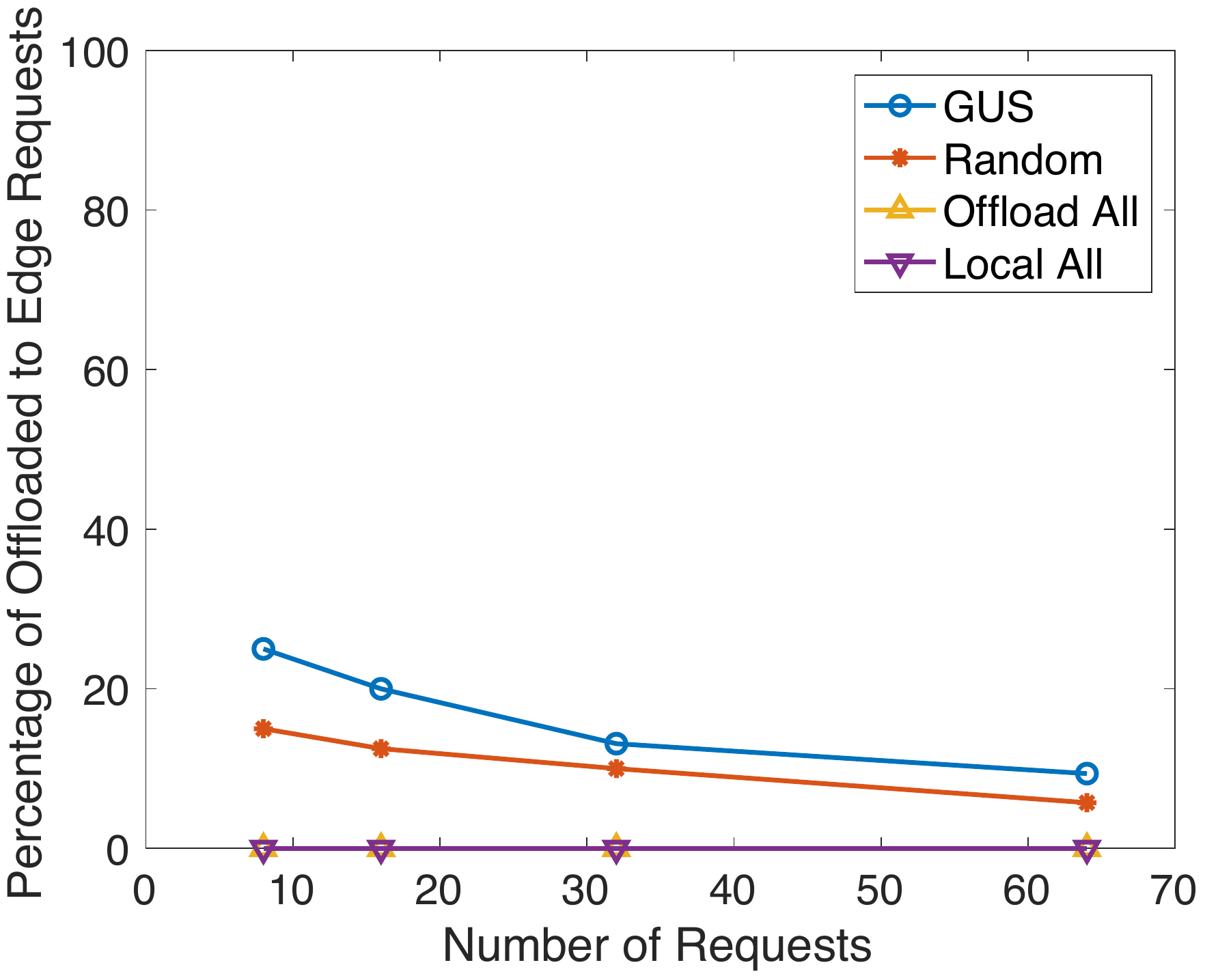}
        \caption{} \label{fig:Del-US-10-5}
    \end{subfigure}
    \caption{\small Performance Evaluation Results: Numerical Results (a)-(d); Testbed Results (e)-(h)} 
    \label{fig:Results}
\end{figure*}

% \begin{figure*}[t!]
%     \includegraphics[width=\linewidth]{minoo_plots}
%     \caption{\small Performance Evaluation Results: Numerical Results (a)-(d); Testbed Results (e)-(h)} 
%     \label{fig:Results}
% \end{figure*}

%\textbf{Comparison with optimal solution.}
We first validated the performance of our proposed algorithm by comparing the
results obtained by GUS and that of the optimal solver (computed by IBM ILOG CPLEX v 12.10.0.0 ) for small test cases. Our results confirm that the proposed algorithm performs close-to-optimal solution (achieving in average $90\%$ of the optimal value) in terms of overall user satisfaction for the scenarios we have tested. The results of such comparison is omitted due to space.%the lack of space.
%We find the optimal solution for a small number of requests in a fixed network setup consisting of two edges and one cloud server which all are connected together. We used IBM ILOG CPLEX (version 12.10.0.0) for solving the model. 
%The number of total requests in system varies between $5$ and $65$ and is assumed to be distributed uniformly among all edge servers.
%The requested completion time, requested accuracy, pre-processing delay, and queue delay are all generated according to random distributions.
%Additionally, we used MATLAB R2019b for implementing the proposed algorithms in the mentioned network setup.
%Our results show that the GUS performs very close, near to $90\%$ to the optimal solution of satisfied users for the scenario we tested. 

\textbf{Baseline algorithms.} We also provided five baseline algorithms for making a comparison between our proposed greedy algorithm and other possible solutions: {\textit{1)~Random-Assignment}}
where one of the servers~(edge/cloud servers) would be selected randomly. If it can satisfy the user requirements and there is enough capacity, it will serve the request; otherwise, the request will be dropped. {\textit{2)~Offload-All}} where sends all the requests from edge servers to the cloud servers to serve.
{\textit{3)~Local Processing-All}}
that chooses only the local edge server for serving each of the requests.
{\textit{4)~Happy-Computation}} in which we assume that there is no limit on the edge servers for the computation capacity and we relax the computation constraint~\eqref{eq:computation} in the optimization model.
{\textit{5)~Happy-Communication}} where we assume no limit on the edge servers' communication capacity and we relax the communication constraint~\eqref{eq:communication} in the optimization model.

\textbf{Numerical Results.} In the following scenarios, we used nine edge servers and one cloud server for the test. The communication delay between the edge servers comes from different communication delays that we got from our testbed in different situations in which the bandwidth is equal to $600$ bytes/msec on average. The average delay between the cloud server and edge server is based on our testbed. To account for edge servers heterogeneity, we assume that there are three types of edge servers in the system which differ based on their storage, communication, and computation capacities. Additionally, we run each test for $20000$ Monte-Carlo runs and report the average. We set $|N|=100$, $|M|=10$, $|K|=100$, $|L|=10$, and it is assumed that services are randomly placed on the edge servers based on their associated storage capacity. We investigated the effect of changing requested delay and requested accuracy in the proposed system. The processing delay of running DL models is based on our testbed results which is between $950$ to $1300$~(msec) for edge servers, and $300$~(msec) for the cloud server.
The requested accuracy~($A_i$) and the requested delay~($C_i$) are both generated according to a normal distribution of $\cal{N}$(45$\%$,10$\%)$ and $\cal{N}$(1000,4000) msec, respectively. $T^q$ is a random number between $0$ and $50$ generated according to a uniform distribution. The $Max_{as}$ and $Max_{cs}$ in the system are $100\%$ and $12000$~(msec), respectively. We assume $w_{ai}=w_{ci}=1$, so equal weights will be considered for both requested accuracy and requested delay. As shown Fig.~\ref{fig:Results}(a), when the requested delay range increases, the total served requests will increase. The reason is that there are more requests which could be sent to the cloud server. When the requested accuracy increases, the number of satisfied requests decreases. The reason is that there would not be a DL model on the edge server which can provide the accuracy that user asked~(Fig.~\ref{fig:Results}(b)). When the number of requests increases, the satisfied users percent decreases~(Fig.~\ref{fig:Results}(c)). The reason is each EC server has limited capacity.
%\minoo{The reason is that each EC server has limited capacity}. 
Finally, when the requests are received at the edge servers, they will be kept in the queue until making decision and serving them. The average queue delay is a function of queue length. When the queue delay increases, the number of satisfied users decreases~(Fig~\ref{fig:Results}(d)). The reason is that the completion time will be more than the requested delay, and therefore, we may not be able to satisfy users and have to drop them.

\textbf{Testbed Implementation.}
\label{sec:Implementation}
We also evaluate our solutions using a real-world testbed, comprised of several different devices to support resource and device heterogeneity. A Linux desktop (Intel core i5-3470-3.20 GHz, RAM 8GB) serves as the central cloud server. It is connected to a NetGear R6020 router, separated by roughly 6 meters. To imitate the delay between the cloud and edge serves, we connected the router to a Raspberry Pi~(RP) 3B~(quad core, RAM 1GB) as a forwarder between the router and the edge servers --- it is roughly 10 meters away from the NetGear router. We then have two RP 4s~(quad core, RAM 4GB) which serve as our edge servers for our testbed. They are connected to the forwarder and are placed on a different floor than the forwarder in the same building, roughly 11 meters away from the forwarder. The user devices are represented by two RP 3Bs which are physically close to our edge servers, at being placed within 1 meter from the edge servers. The services we consider in our testbed are two pre-trained models: \emph{GoogleNet}~\cite{szegedy2015going} (exclusively available on the cloud) and \emph{SqueezeNet}~\cite{iandola2016squeezenet} (placed on edge servers, but provides poorer accuracy and resource cost). The data we use for submitting service requests for image classification is provided by the ImageNet dataset~\cite{imagenet_cvpr09}. We implement a simple management program in C++ that runs the decision algorithms and makes offloading decisions w.r.t. user requests. Common communication errors had to be addressed as well (but for brevity we do not elaborate on them here).

\textbf{Testbed Results. }
%We used the following parameters for next tests. For requested delay and requested accuracy, if we wanted to use a random number based on normal distribution, we had to repeat each test at least $1000$ times to ensure that we are covering all numbers in a range. We used a mixture of fixed and random number, together.
For our testbed results, users requested thresholds on completion time and accuracy are set at $C_i=53000$~(msec) and the $A_i=50\%$ for all requests. We used a fixed number equal to one for both the $w_{ai}$ and the $w_{ci}$ $\forall i \in N$, a fixed length of $4$ for the queue, and the length of each time frame to $3000$~(msec), which is essentially representing how often we run a decision algorithm if the queue is not full.
%As mentioned, either when the queue is full or a time frame is finished, we run the decision making algorithm. 
We used multi-threading on the edge servers and we set the processing capacity equal to three as the maximum number of threads that can run the image classification in each time-step. The communication capacity is equal to ten for each edge server. This means that each edge server can only send $10$ images to other servers in a time slot. We repeated each test for two hours to remove the effects of randomness due to the wireless channel and averaged the results. When running our decision algorithms, we compute the expected processing delay based on the data we first collected from our testbed. As mentioned, the processing delay of running SqueezeNet on Raspberry 4 is $1300$~(msec), on average. And, the processing delay of running GoogleNet on the cloud server is equal to $300$~(msec).
Given the dynamics of the real-world scenarios and the fact that the wireless channel is changing all the time, we update the value of the expected communication delay used in our GUS algorithm according to the following simple rule. The algorithm starts with an initial estimated value for the expected communication delay based on our previous observations and historical data. At each iteration, this value is updated based on the observations made from the previous round. 
Given the expected communication bandwidth $B$ at time $t$, we compute the new expected communication bandwidth at time $t+1$ as $E[B_{{t+1}}] = \frac{B_{t} + B_{{t-1}}}{2}$.
Using the expected bandwidth, we can then compute the expected communication delay for each image based on the expected bandwidth and the size of the image. For our tests, we started with $B=600$ bytes/msec for the expected communication bandwidth between the edge and the cloud. We may also have to adapt the $Max_{cs}$ parameter given that the value of expected completion time is updated. 
%If the communication delay increases and completion time gets bigger than $Max_{cs}$, we change the value of $Max_{cs}$ into the new (updated) worst completion time.
The results of the testbed implementation are shown in Fig~\ref{fig:Results}. We increased the total number of requests sent to the EC system and observed the performance of GUS, and three heuristics namely, ``random", ``local all" and ``offload all" heuristics. Fig~\ref{fig:Results}(e) shows the average user satisfied percentage proving that GUS is always providing a much better rate of US compared to the other two heuristics. In the case of the ``offload all" and ``local all" heuristic, by increasing the number of requests, the percentage of satisfied users is decreasing because of the bottleneck created by the communication and computation capacity of the edge servers, respectively. This shows that an optimal combination of local processing and offloading (as provided by GUS) can significantly increase the US of the system. Fig.~\ref{fig:Results}(f) shows the percentage of requests locally processed, Fig.~\ref{fig:Results}(g) shows the percentage of the requests offloaded to the cloud, and Fig.~\ref{fig:Results}(h) shows the percentage of requests that are offloaded to the other edge servers.
Comparing the results provided in Fig~\ref{fig:Results}, 
%we can say that GUS provides, on average, satisifes $50\%$ users than the heuristics.
we can say that the percentage of the users satisfied using GUS is on average $50\%$ more than that of the heuristics.

\section{Conclusion}
In this paper, we proposed a new user satisfaction~(US) metric that considers the trade-off between accuracy and delay of the provided service for DL services in EC systems. We pose an offloading optimization model to maximize the proposed US given a set of capacity constraints on the edge servers and proved that this is NP-hard. Hence, we came up with a greedy algorithm, called GUS, to solve the model in polynomial-time. Finally, we evaluate on a real EC testbed consisting of devices of varying resources. 
The implementation results also prove the effectiveness of our proposed algorithm in comparison with baseline heuristics.
Our future work will focus on considering different priorities for the requests and impacts of user mobility on the provided trade-off.
% Additionally, we assumed that we know users' requested delay and requested accuracy. The approach of predicting a good value of requested delay and requested accuracy will be investigated in future work.

\section*{Acknowledgements}
\noindent This work is funded by research grants provided by the Cisco Systems Inc. and the National Science Foundation (NSF) under the grant numbers 1215519250 and 1948387 respectively.
\bibliography{refs}
\bibliographystyle{ieeetr}
\end{document}